\documentclass[a4paper, 11pt, DIVcalc, oneside, headsepline, ngerman, openany, bibtotoc, smallheadings,liststotoc]{article}
\usepackage[T1]{fontenc}
\usepackage[utf8,ansinew]{inputenc}
\usepackage{scrpage2}
\usepackage{makeidx}
\makeindex
\usepackage{enumerate}

\usepackage{mathtools}
\usepackage{amssymb,multicol,multirow}
\usepackage{amsthm}
\usepackage[T1]{fontenc}
\usepackage[all]{xy}
\usepackage{verbatim}
\usepackage{hhline}
\newtheorem{theorem}{Theorem}

\newtheorem{proposition}{Proposition}
\newtheorem{definition}{Definition}
\newtheorem{corollary}{Corollary}

\usepackage[pdftex]{graphicx}
\usepackage[intoc]{nomencl}

\makenomenclature
\usepackage[numbers]{natbib}
\usepackage{pdfpages}
\usepackage[vcentering,dvips]{geometry}
\usepackage{color}
\usepackage{url}
\usepackage{hyperref}
\hypersetup{colorlinks=true, linkcolor=blue,urlcolor=blue, citecolor=blue}
\usepackage{caption}
\usepackage{fancybox}
\usepackage{longtable}
\usepackage{tabularx}
\usepackage{eurosym}
\newcolumntype{C}[1]{>{\centering\arraybackslash}p{#1}}
\newcolumntype{R}[1]{>{\raggedleft\arraybackslash}p{#1}}

\title{Exploiting Equitable Partitions for Efficient Block Triangularization}
\author{Mario Thüne}
\date{}
\begin{document}
\maketitle
\begin{abstract}
In graph theory a partition of the vertex set of a graph is called equitable if for all pairs of cells all vertices in one cell have an equal number of neighbours in the other cell. Considering the implications for the adjacency matrix one may generalize that concept as a block partition of a complex square matrix s.t. each block has constant row sum. It is well known that replacing each block by its row sum yields a smaller matrix whose multiset of eigenvalues is contained in the initial spectrum. We generalize this approach to weighted row sums and rectangular matrices and derive an efficient unitary transformation which approximately block triangularizes a matrix w.r.t. an arbitrary partition. Singular values and Hermiticity (if present) are preserved. The approximation is exact in the equitable case and the error can be bounded in terms of unitarily invariant matrix norms.
\end{abstract}

\section{Introduction}
\subsection{Equitable Partitions}
Let $\Gamma$ be a (multi-)graph and let $\mathbf{A}$ be its adjacency matrix, whose entries $a_{vw}$ are the number of edges connecting vertices $v$ and $w$.
Let $\Pi=\left(c_1,\ldots,c_k\right)$ be a partition of the vertex set of $\Gamma$ into $k$ cells, inducing a block partition of $\mathbf{A}$, i.e. a simultaneous (disjoint and exhaustive) partition of its rows and columns. It is convenient to define an \emph{indicator matrix} of a partition as
\begin{definition}\label{defB}
$$\mathbf{B}=\left(b_{vi}\right)\in\left\{0,1\right\}^{N\times k}\text{ with }b_{vi}=\left\{\begin{array}{ll}1&\text{, if item $v$ is in cell $i$}\\0&\text{, else.}\end{array}\right.$$
\end{definition}
The partition $\Pi$ is called equitable if all vertices of $\Gamma$ in the same cell have the same number of neighbours in any cell. Equivalently, we may call it equitable if each induced submatrix of $\mathbf{A}$ has constant row sum. The equitable partitions of $\mathbf{A}$ ordered by refinement form a lattice which contains the trivial equitable partition, in which every cell has size exactly one, as the minimum. From the definition it follows that a partition of $\mathbf{A}$ is equitable if and only if there exits a matrix $\mathbf{\Theta}=\left(\theta_{ij}\right)$ s.t.
\begin{equation}\label{equiCentral}
\mathbf{A}\mathbf{B}=\mathbf{B}\mathbf{\Theta}
\quad\text{i.e.}\quad
\forall\ i,j\in\left\{1,\ldots,k\right\}\ \forall\ u\in\left\{1,\ldots,n_i\right\}\ \sum\limits_v^{n_j}\mathbf{A}_{ij,{uv}}=\theta_{ij}.
\end{equation}
The matrix $\mathbf{\Theta}$ is called the quotient of the partition. Its entries $\theta_{ij}$ are the constant row sums of the matrix blocks $A_{ij}$ induced by the cells $c_i$ and $c_j$, which are the number of edges connecting a fixed vertex in $c_i$ to vertices in $c_j$.
\subsection{Applications in Graph and Matrix Theory}
The notion of equitable partitions was developed in graph theory. In network analysis the same concept is also known as \emph{exact coloration}~\cite{EB94RegEqu} or \emph{exact role assignment}~\cite{lerner2005assignments}. It is closely related to \emph{graph fibration}~\cite{BV02FibGra} and arises naturally in the context of graph automorphisms problems since every non trivial automorphism induces a non trival equitable partition. As graph invariants which can be searched for quickly using quite efficient algorithms, equitable partitions are useful in attacking graph isomorphism problems. In that context they are also known as  \emph{1-dimensional Weisfeiler-Lehman stabilizers}~\cite{CFI92OptLow}.\newline
Block partitioned matrices s.t. each block has constant row sum are called \emph{block-stochastic matrices}. In the context of markov chains the technique of \emph{lumping} exploits equitable partitions in order to reduce the number of states~\cite{Buchholz94}. The quotient $\mathbf{\Theta}$ is also known as the \emph{front divisor}~\cite{CDS80SpeGra}. Famously, the spectrum of $\mathbf{\Theta}$, called the main spectrum, is a subset of the spectrum of $\mathbf{A}$ since the columns of $\mathbf{B}$ span an invariant subspace if \eqref{equiCentral} holds. Therefore, there is a similarity transformation which $2\times2$ block triangularizes $\mathbf{A}$ s.t. one diagonal block is the quotient. Such a transformation can be constructed and applied efficiently in a way provided in~\cite{H59AppThe},~\cite{Chang2011559}. The block triangularization method given below differs from that approach in order to fit in a generalized framework of equitability and provides efficient unitary transformations.
\subsection{Aim and Outline}
We will generalize the notion of ordinary equitable partitions to arbitrary weighted partitions of the rows and columns of complex matrices. According to a given partition we derive an efficient and stable unitary similarity transformation in order to $2\times2$ block triangularize the matrix up to an error term, which is minimized w.r.t. to several matrix norms and vanishes if and only if exact equitability holds. The transformation can be computed in $O\left(N\right)$ and applied in $O\left(N^2\right)$. It can be further generalized enabling the application to rectangular matrices while maintaining the unitarity property. However the further generalized transformation does only preserve the singular values, but (in general) not the spectrum.\newline
Despite offering insides into the structure of objects represented by a graph or a matrix our notion of equitability and its corresponding transformation may be used for compression and for preprocessing eigen and singular value problems. Although describing our transformation as an efficient compression method may seem to suggest that the exploited structure is, in a sense, rare, the concept of equitable partitions, as indicated above, is rather common in various applications, where it is found directly in the studied problem or as an interesting exceptional or ideal case. The usefulness might be increased in particular by the fact that deviations from an exact equitability may be allowed within our framework.
\newline
In order to get used to the concept and some notation, we briefly discuss in section~\eqref{secUEP} the special case of an ordinary unweighted equitable partition of a complex square matrix including the derivation of the associated efficient unitary block triangularization and we give an example. In the main part, section~\eqref{secWEP}, we consider weighted not necessarily equitable partitions introducing the deviation matrix and give our main theorem. In section~\eqref{secDMR} we consider non exact equitability as an eigenvalue perturbation, give a short overview of several other known generalizations of equation~\eqref{equiCentral}, and briefly consider the problem of finding an equitable partition. Our further generalized version of the concept applicable to rectangular matrices can be found in the appendix.\newline
Throughout the article we use the apostrophe to denote the complex conjugated transpose without distinguishing between real and complex operands and we utilize the following notation
\begin{definition} Let $n\in\mathbb{N}$.
$\mathbf{j}_n=(\hspace{0.1em}\underbrace{1,\ldots,1}_{\text{$n$ times}}\hspace{0.1em})'$
and
$\mathbf{f}_n=(1,\hspace{-0.5em}\underbrace{0,\ldots,0}_{\text{$(n-1)$ times}}\hspace{-0.5em})'$.
\end{definition}

\section{Unweighted Equitable Partitions}\label{secUEP}
\subsection{Indicator Matrix and Quotient}
Let $\mathbf{A}\in\mathbb{C}^{N\times N}$ and let $\Pi=\left(c_1,\ldots,c_k\right)$ be a simultaneous (disjoint and exhaustive) partition of its rows and columns with \emph{indicator matrix} $\mathbf{B}$ as in definition~\eqref{defB}. Let $\mathbf{A}_{ij}\in\mathbf{C}^{n_i\times n_j}$ be the matrix block in $\mathbf{A}$ induced by row cell $c_i$ and column cell $c_j$. Let $n_i$ be the size of the cell $c_i$ and let
\begin{equation}
\mathbf{N}=\left(\mathbf{B}'\mathbf{B}\right)^{\frac{1}{2}}=\operatorname{diag}\left(\sqrt{n_1},\ldots,\sqrt{n_k}\right)
\end{equation}
We introduce the front quotient, the rear quotient and the Rayleigh quotient respectively as
\begin{equation}
\mathbf{E}^{-}=\mathbf{N}^{-2}\mathbf{B}'\mathbf{A}\mathbf{B},\quad
\mathbf{E}^{+}=\mathbf{B}'\mathbf{A}\mathbf{B}\mathbf{N}^{-2},\quad
\mathbf{E}^{\mathrm{0}}=\mathbf{N}^{-1}\mathbf{B}'\mathbf{A}\mathbf{B}\mathbf{N}^{-1}
\end{equation}
We call $\mathbf{A}$ \emph{front equitable} (i) and respectively \emph{rear equitable} (ii) w.r.t. $\mathbf{B}$ if
\begin{equation}\label{FReq_part}
\left(i\right)\ \mathbf{A}\mathbf{B}=\mathbf{B}\mathbf{E}^{-}\quad,\quad
\left(ii\right)\ \mathbf{B}'\mathbf{A}=\mathbf{E}^{+}\mathbf{B}'.
\end{equation}
It is easy to see that for Hermitian matrices row equitability and column equitability imply each other. For the rest of this section we assume front equitability, i.e. 
\begin{equation}\label{fronteq}
\mathbf{A}_{ij}\mathbf{j}_{n_j}=e^{-}_{ij}\mathbf{j}_{n_i}\quad \forall\ i,j\in\left\{1,\ldots,k\right\}.
\end{equation}
\subsection{Block Triangularization}
In order to block triangularize $\mathbf{A}$ according to $\mathbf{B}$ we utilize the Householder matrices
\begin{equation}\label{EPhouseholder}
\mathbf{H}_i=\mathbf{I}_{n_i}-2\frac{\mathbf{y}_i\mathbf{y}_i'}{\mathbf{y}_i'\mathbf{y}_i}\ ,\quad
\mathbf{y}_i=\mathbf{j}_{i}+\sqrt{n_i}\mathbf{f}_{n_i}.
\end{equation}
The following useful relations are easily verified
\begin{equation}\label{EPrelations}
\mathbf{H}_i\mathbf{f}_{n_i}=-\frac{1}{\sqrt{n_i}}\mathbf{j}_{n_i}\quad,\quad
\mathbf{H}_i'\mathbf{j}_{n_i}=-\sqrt{n_i}\mathbf{f}_{n_i}.
\end{equation}
In order to simplify notations but w.l.o.g. we assume \emph{suitable indexing} which means that $\mathbf{A}$ and $\mathbf{B}$ are indexed in such a way that for $u$ in cell $c_i$ and $v$ in cell $c_j$ it holds that $i<j$ implies $u<v$. Then our proposed transformation of $\mathbf{A}$ can be written conveniently in matrix form using the matrix
\begin{equation}
\mathbf{\tilde{H}}=\operatorname{diag}\left(\mathbf{H}_1,\ldots,\mathbf{H}_k\right),
\end{equation}
which is explicitly block diagonal and, according to \eqref{EPhouseholder}, unitary.
\begin{equation}\label{transformA}
\mathbf{\tilde{A}}=\mathbf{\tilde{H}}'\mathbf{A}\mathbf{\tilde{H}}
=\left(\begin{array}{ccc}\mathbf{\tilde{A}}_{11}&\cdots&\mathbf{\tilde{A}}_{1k}\\
\vdots&\ddots&\vdots\\
\mathbf{\tilde{A}}_{k1}&\cdots&\mathbf{\tilde{A}}_{kk}
\end{array}\right)
\quad\text{with}\quad
\mathbf{\tilde{A}}_{ij}=\mathbf{H}_{i}'\mathbf{A}_{ij}\mathbf{H}_{j}.
\end{equation}
By \eqref{EPrelations} and \eqref{fronteq} there exists a matrix $\mathbf{E}=\left(e_{ij}\right)$ s.t.
\begin{equation}
\mathbf{\tilde{A}}_{ij}\mathbf{f}_{j}=e_{ij}\mathbf{f}_{n_i},
\end{equation}
which immediately shows that each $\mathbf{\tilde{A}}_{ij}$ is block triangular with the left upper block being the scalar $e_{ij}$. Therefore, there is a readily available, in general not unique permutation matrix $\mathbf{\Omega}$ such that
\begin{equation}
\mathbf{\hat{A}}=\mathbf{\Omega}'\mathbf{\tilde{A}}\mathbf{\Omega}=
\left(\begin{array}{cc}\mathbf{E}&\mathbf{D}\\\mathbf{0}&\mathbf{F}\end{array}\right)
\end{equation}
is explicitly block triangular. Since the applied transformations are unitary, the spectrum and the singular values of $\mathbf{A}$ are preserved. We will refer to $\mathbf{F}$, which in general depends on the indexing of $\mathbf{A}$ and on $\mathbf{\Omega}$, as a \emph{factor}. One shows that all factors are unitarily equivalent and that by similarity
\begin{equation}\sigma\left(\mathbf{A}\right)=\sigma\left(\mathbf{E}\right)+\sigma\left(\mathbf{F}\right).
\end{equation}
Additionally, if $\mathbf{v}$ is an eigenvector of $\mathbf{\hat{A}}$ then $\mathbf{\tilde{H}}\mathbf{\Omega}\mathbf{v}$ is an eigenvector of $\mathbf{A}$ to the same eigenvalue. One also shows that $\mathbf{D}$ vanishes if and only if rear equitability holds. The computational costs for the transformation $\mathbf{\tilde{H}}$ are of order $O\left(n_in_j\right)$ on each subblock for we apply only matrix vector multiplication and matrix addition since $\mathbf{H}_i$ is a rank one update of the identity. Therefore, the total costs are of order $O\left(N^2\right)$. Since $\mathbf{\tilde{H}}\mathbf{\Omega}$ is unitary, Hermiticity (if present) of $\mathbf{A}$ is preserved. Numeric stability is supported by using Householder matrices. Note that in this section we constructed $\mathbf{\tilde{H}}$ s.t. $\mathbf{E}=\mathbf{E}^{\mathrm{0}}$. In the general case those two matrices are unitarily equivalent but not necessarily identical.
\subsection{Example}
Let
$$\mathbf{A}_0=\left(\begin{array}{cccccc}
1&2&3&3&3&2\\
2&4&3&1&2&1\\
3&3&1&4&1&1\\
3&1&4&0&2&3\\
3&2&1&2&3&2\\
2&1&1&3&2&4
\end{array}\right)\quad\text{and}\quad
\mathbf{P}_0=\left(\begin{array}{cccccc}
1&0&0&0&0&0\\
0&1&0&0&0&0\\
0&0&0&0&0&1\\
0&0&0&1&0&0\\
0&0&0&0&1&0\\
0&0&1&0&0&0
\end{array}\right)
$$
One verifies that $\mathbf{A}_0$ is (unweighted) front equitable w.r.t. $\Pi_0=\left(1\vert2,6\vert3,4,5\right)$. Using the permutation $\mathbf{P}_0$
we can transform it into the suitably indexed form
$$\mathbf{A}=\mathbf{P}_0'\mathbf{A}_0\mathbf{P}_0=\left(\begin{array}{cccccc}
1&2&2&3&3&3\\
2&4&1&1&2&3\\
2&1&4&3&2&1\\
3&1&3&0&2&4\\
3&2&2&2&3&1\\
3&3&1&4&1&1
\end{array}\right),$$
which is (unweighted) front equitable w.r.t. $\Pi=\left(1\vert2,3\vert4,5,6\right)$ with front quotient $$\mathbf{E}^{-}=\left(\begin{array}{ccc}
1&4&9\\
2&5&6\\
3&4&6
\end{array}\right).$$
One may employ
$$\mathbf{H}_1=\mathbf{H}\left(\mathbf{j}_1\right)=-1,$$
$$\mathbf{H}_2=\mathbf{H}\left(\mathbf{j}_2\right)=-\frac{1}{\sqrt{2}}\left(\begin{array}{cc}1&1\\1&-1\end{array}\right),$$ 
$$\mathbf{H}_3=\mathbf{H}\left(\mathbf{j}_3\right)=-\frac{1}{\sqrt{3}}\left(\begin{array}{ccc}1&1&1\\1&\frac{1+\sqrt{3}}{-2}&\frac{1-\sqrt{3}}{-2}\\1&\frac{1-\sqrt{3}}{-2}&\frac{1+\sqrt{3}}{-2}\end{array}\right)$$ and 
$\mathbf{\tilde{H}}=\operatorname{diag}\left(\mathbf{H}_1,\mathbf{H}_2,\mathbf{H}_3\right)$
to transform $\mathbf{A}$ s.t.
$$\mathbf{\tilde{A}}=\mathbf{\tilde{H}}'\mathbf{A}\mathbf{\tilde{H}}=\left(\begin{array}{cccccc}
1&\frac{4}{\sqrt{2}}&0&\frac{9}{\sqrt{3}}&0&0\\
\frac{4}{\sqrt{2}}&5&0&6\frac{\sqrt{2}}{\sqrt{3}}&0&0\\
0&0&3&0&\mbox{-}3\scalebox{0.8}{\mbox{+}}\sqrt{3}&\mbox{-}3\mbox{-}\sqrt{3}\\
\frac{9}{\sqrt{3}}&6\frac{\sqrt{2}}{\sqrt{3}}&0&6&0&0\\
0&0&\mbox{-}3\scalebox{0.8}{\mbox{+}}\sqrt{3}&0&\sqrt{3}\mbox{-}1&\mbox{-}6\\
0&0&\mbox{-}3\mbox{-}\sqrt{3}&0&\mbox{-}6&\mbox{-}\sqrt{3}\mbox{-}1
\end{array}\right).$$
Using the permutation $\mathbf{\Omega}=\left(\begin{array}{cccccc}
1&0&0&0&0&0\\
0&1&0&0&0&0\\
0&0&0&1&0&0\\
0&0&1&0&0&0\\
0&0&0&0&1&0\\
0&0&0&0&0&1
\end{array}\right)$ we obtain the matrix
$$\mathbf{\hat{A}}=\mathbf{\Omega}'\mathbf{\tilde{A}}\mathbf{\Omega}=\left(\begin{array}{cccccc}
1&\frac{4}{\sqrt{2}}&\frac{9}{\sqrt{3}}&0&0&0\\
\frac{4}{\sqrt{2}}&5&6\frac{\sqrt{2}}{\sqrt{3}}&0&0&0\\
\frac{9}{\sqrt{3}}&6\frac{\sqrt{2}}{\sqrt{3}}&6&0&0&0\\
0&0&0&3&\mbox{-}3\scalebox{0.8}{\mbox{+}}\sqrt{3}&\mbox{-}3\mbox{-}\sqrt{3}\\
0&0&0&\mbox{-}3\scalebox{0.8}{\mbox{+}}\sqrt{3}&\sqrt{3}\mbox{-}1&\mbox{-}6\\
0&0&0&\mbox{-}3\mbox{-}\sqrt{3}&\mbox{-}6&\mbox{-}\sqrt{3}\mbox{-}1
\end{array}\right),$$
which is explicitly reducible. Since $\mathbf{A}$ is Hermitian, the (unweighted) partition $\Pi$ induces front and rear equitability and we actually obtain a block diagonal form. Note that both blocks are Hermitian but the front quotient $\mathbf{E}^{-}$ is not. One verifies that $$\mathbf{E}=\left(\begin{array}{ccc}
1&\frac{4}{\sqrt{2}}&\frac{9}{\sqrt{3}}\\
\frac{4}{\sqrt{2}}&5&6\frac{\sqrt{2}}{\sqrt{3}}\\
\frac{9}{\sqrt{3}}&6\frac{\sqrt{2}}{\sqrt{3}}&6
\end{array}\right)=\operatorname{diag}\left(1,2,3\right)^{\frac{1}{2}}\mathbf{E}^{-}\operatorname{diag}\left(1,2,3\right)^{-\frac{1}{2}}.$$
Let $\mathbf{F}$ denote the lower diagonal block of $\mathbf{\hat{A}}$. Let $\mathbf{V}_{\mathbf{E}}$ and $\mathbf{V}_{\mathbf{F}}$ be the eigenvector matrices of $\mathbf{E}$ and $\mathbf{F}$, respectively. Then one shows that $$\mathbf{V}=\mathbf{P}_0\mathbf{\tilde{H}}\mathbf{\Omega}\left(\begin{array}{cc}\mathbf{V}_{\mathbf{E}}&\mathbf{0}\\\mathbf{0}&\mathbf{V}_{\mathbf{F}}\end{array}\right)$$ is an eigenvector matrix of $\mathbf{A}$. Note that $\mathbf{A}$ and $\mathbf{V}$ need more storage than $\mathbf{\hat{A}}$, $\mathbf{V}_{\mathbf{E}}$, and $\mathbf{V}_{\mathbf{F}}$. The transformations $\mathbf{P}_0$, $\mathbf{\tilde{H}}$ and $\mathbf{\Omega}$ follow from $\Pi_0$ which can be stored as a vector. Due to the small size the blocks of $\mathbf{\tilde{H}}$ were given explicitly as dense matrices. For larger problems one would prefer the usual sparse form as a rank one update of the identity given in \eqref{EPhouseholder}.
\section{Weighted Equitable Partitions}\label{secWEP}
\subsection{Preliminaries}
In this section we generalize equitable partitions and accordingly the proposed block triangularization method for square matrices. We introduce the generalized quotient defined for arbitrary partitions of a matrix as a generalization of front and rear quotient. We also introduce the deviation vectors and the deviation matrix and utilize the norm of the latter in order to quantify deviations of a given partition from our generalized notion of equitability. The generalization of the efficient unitary similarity transformation introduced above yields a block triangularization up to an error term due to the deviation from equitability. A further generalization applicable to rectangular matrices preserving only singular values but in general not the spectrum is discussed in the appendix.\newline
Note that whenever we invert a matrix explicitly (i.e. not by complex conjugated transposition) this matrix is diagonal. The occasional uses of the pseudo inverse with the property
\begin{equation}
c^{\dagger}=\left\{\begin{array}{cc}
0&,c=0\\ \frac{1}{c}&,\text{else}
\end{array}\right. ,\quad c\in\mathbb{C}
\end{equation}
may be regarded as merely technical.
\subsection{Complex Householder Transformations}
This subsubsection aims at the transformation in definition~\eqref{def_H(x)} and its properties given in~\eqref{eq_relations_H}. We consider elementary unitary matrices (EUMs) which are rank (at most) one updates of the identity and necessarily (in order to be unitary)~\cite{S95EleUni} of the form
\begin{equation}
\mathbf{U}\left(\gamma,\mathbf{y}\right)=\mathbf{I}-\frac{2}{1+i\gamma}\left(\mathbf{y}'\mathbf{y}\right)^{\dagger}\mathbf{y}\mathbf{y}',\quad\mathbf{y}\in\mathbb{C}^{n},\gamma\in\mathbb{R}.
\end{equation}
EUMs are a complex generalization of real Householder matrices~\cite{H58UniTri},~\cite{L96ComEle}. We observe that for $c\in\mathbb{C}\setminus\left\{0\right\}$
and $\mathbf{P}$ being a permutation matrix
\begin{equation}\label{eqEscale}
\mathbf{U}\left(\gamma,c\mathbf{y}\right)=\mathbf{U}\left(\gamma,\mathbf{y}\right)\ ,\quad
\mathbf{U}\left(\gamma,\mathbf{P}\mathbf{y}\right)=\mathbf{P}\mathbf{U}\left(\gamma,\mathbf{y}\right)\mathbf{P}'.
\end{equation}
Let $\mathbf{x}$ and $\mathbf{z}$ be non vanishing complex vectors. We seek an EUM mapping $\mathbf{z}$ into the direction of $\mathbf{x}$, i.e. a complex vector $\mathbf{y}$ and a real number $\gamma$ s.t.
\begin{equation}\label{eqEytox}
\mathbf{U}\left(\gamma,\mathbf{y}\right)\mathbf{z}=\alpha\mathbf{x}\quad\text{with}\quad\alpha\in\mathbb{C}\setminus\left\{0\right\},
\end{equation}
which implies that $\lVert\mathbf{x}\rVert$ and $\lVert\mathbf{z}\rVert$ determine $\alpha$ up to a phase factor
\begin{equation}
\sqrt{\left(\mathbf{U}\left(\gamma,\mathbf{y}\right)\mathbf{z}\right)
'\left(\mathbf{U}\left(\gamma,\mathbf{y}\right)\mathbf{z}\right)}=
\lVert\mathbf{z}\rVert=\left|\alpha\right|\lVert\mathbf{x}\rVert.
\end{equation}
Again by~\eqref{eqEytox} $\mathbf{y}$ is a linear combination of $\mathbf{x}$ and $\mathbf{z}$, namely 
\begin{equation}\label{u=y+ax=cu}
\mathbf{z}-\alpha\mathbf{x}=\frac{2}{1+i\gamma}\left(\mathbf{y}'\mathbf{y}\right)^{\dagger}\left(\mathbf{y}'\mathbf{z}\right)\mathbf{y}.
\end{equation}
Since according to \eqref{eqEscale} scaling of $\mathbf{y}$ does not change $\mathbf{U}\left(\gamma,\mathbf{y}\right)$ we may choose $\mathbf{y}=\mathbf{x}-\frac{1}{\alpha}\mathbf{z}$.
We are particular interested in the case $\mathbf{z}=\mathbf{f}_n$. Setting $\alpha=\beta\left|\alpha\right|$ and using \eqref{u=y+ax=cu}, we reach in the non trivial case, $\mathbf{y}\neq0$,
\begin{equation}\label{eqGamma}
\gamma\left(\mathbf{x},\beta\right)=
\left(\lVert\mathbf{x}\rVert-\operatorname{Re}\left(\beta x^1\right)\right)^{\dagger}\operatorname{Im}\left(\beta x^1\right).
\end{equation}
Thus, the required EUM of $\mathbf{f}_n$ into the direction of $\mathbf{x}$ is determined up to a complex parameter $\beta$ lying on the unit circle. We introduce \begin{equation}\label{H=U}
\mathbf{H}\left(\mathbf{x},\beta\right)=
\mathbf{U}\left({\gamma\left(\mathbf{x},\beta\right)},
\mathbf{x}-\lVert\mathbf{x}\rVert\overline{\beta}\mathbf{f}_n\right)\text{ with }
\left|\beta\right|=1
\end{equation}
and give an explicit definition.
\begin{definition}\label{def_H(x)}
Let $\mathbf{x}\in\mathbb{C}^n$ with $\lVert\mathbf{x}\rVert>0$, let $x^1=\mathbf{f}_n'\mathbf{x}$ denote its first entry and let $\beta$ be a complex number with $\left|\beta\right|=1$, then
$$\mathbf{H}\left(\mathbf{x},\beta\right)=\left\{\begin{array}{cc}
\mathbf{I}_n & ,\frac{1}{\lVert\mathbf{x}\rVert}\mathbf{x}=\overline{\beta}\mathbf{f}_n\\
\mathbf{I}_n+\frac{
\left(\mathbf{x}-\lVert\mathbf{x}\rVert\overline{\beta}\mathbf{f}_n\right)
\left(\mathbf{x}-\lVert\mathbf{x}\rVert\overline{\beta}\mathbf{f}_n\right)'}
{\lVert\mathbf{x}\rVert\overline{\beta}\left(
\overline{x^1}-\lVert\mathbf{x}\rVert\beta\right)} & ,\text{ else. }
\end{array}\right.
$$
\end{definition} 
Using $\mathbf{y}=\mathbf{x}-\lVert\mathbf{x}\rVert\overline{\beta}\mathbf{f}_n$ we may rewrite
\begin{equation}
\mathbf{H}\left(\mathbf{x},\beta\right)=\mathbf{I}_n+
\frac{\beta}{\lVert\mathbf{x}\rVert}\left(\mathbf{y}'\mathbf{f}_n\right)^{\dagger}
\mathbf{y}\mathbf{y}'=\mathbf{I}_n-\left(\mathbf{x}'\mathbf{y}\right)^{\dagger}\mathbf{y}\mathbf{y}'.
\end{equation}
And we summarize the following properties
\begin{equation}\label{eq_relations_H}
\mathbf{H}\left(\mathbf{x},\beta\right)\mathbf{f}_n
=\frac{\beta}{\lVert\mathbf{x}\rVert}\mathbf{x}\quad
\text{and}\quad\mathbf{H}\left(\mathbf{x},\beta\right)'\mathbf{x}
=\frac{\lVert\mathbf{x}\rVert}{\beta}\mathbf{f}_n.
\end{equation}
Since $\mathbf{H}\left(\mathbf{x},\beta\right)$ is a rank one update of $\mathbf{I}_n$, it can be stored with $O\left(n\right)$ and multiplied with a square matrix of size $n$ in $O\left(n^2\right)$. Note that $\mathbf{H}\left(\mathbf{x},\beta\right)$ crucially depends on the ordering of the entries of $\mathbf{x}$,
\begin{equation}\label{eq_permdependent_H}
\mathbf{H}\left(\mathbf{P}'\mathbf{x},\beta\right)\neq\mathbf{P}'\mathbf{H}\left(\mathbf{x},\beta\right)\mathbf{P}\quad\text{for general $\mathbf{x}$ and permutation matrix $\mathbf{P}$}.
\end{equation}
Although its norm is determined to be $1$, the actual choice of $\beta$ is arbitrary. We may exploit that freedom in order to enhance the numerical properties of $\mathbf{H}\left(\beta,\mathbf{x}\right)$. Particular useful is a choice s.t. $\beta x^1\in\mathbb{R}$, implying $\gamma=0$ by \eqref{eqGamma} and leading to a Hermitian matrix. Furthermore, for real $\mathbf{x}$, $\beta\in\left\{-1,1\right\}$ ensures a real matrix. A practical recommendation might be
\begin{definition}\label{beta0}
$\beta_0\left(\mathbf{x}\right)=\left\{\begin{array}{ll}
-\frac{\overline{x^1}}{\left|x^1\right|}&, x^1\neq0\\
1&, x^1=0\end{array}\right.\quad,\ \mathbf{x}\in\mathbb{C}^n$,
\end{definition}\noindent
which supports numerical stability and coincides with the usual recommendation for the numerical construction of a real Householder matrix. In the previous section we applied $\beta_0$ tacitly.\newline
\subsection{Weighted Partition, Quotient and Deviation Matrix}
Let $\Pi=\left(c_1,\ldots,c_k\right)$ be a partition of $\{1,\ldots,N\}$ into $k$ cells with indicator matrix $\mathbf{B}$. Let $\mathbf{w}\in\mathbb{C}^{N}$ and $\mathbf{A}\in\mathbb{C}^{N\times N}$. We introduce the \emph{weighted indicator matrix} \mbox{$\mathbf{W}=\operatorname{diag}\left(\mathbf{w}\right)\mathbf{B}$}.
\begin{definition}\label{defWIM}
Let $\Pi=\left(c_1,\ldots,c_k\right)$ be a partition of $\left\{1,\ldots,N\right\}$. Let $\mathbf{w}\in\mathbb{C}^N$ and let $w^v$ denote its $v$-th entry.
$$\mathbf{W}=\left(w_{vi}\right)\in\mathbb{C}^{N\times k}\text{ with }w_{vi}=\left\{\begin{array}{ll}w^v&\ v\in c_i\\0&\text{, else}\end{array}\right.$$
\end{definition}
A weighted indicator matrix $\mathbf{W}$ is called \emph{admissible} if $\lVert\mathbf{w}_i\rVert$ for all vector blocks $\mathbf{w}_i$ induced by $c_i$. This implies that $\mathbf{W}'\mathbf{W}$ is invertible and ultimately ensures that the complete spectrum of the quotient, to be defined below, is contained in the spectrum of $\mathbf{A}$. For the rest of this section we assume admissibility.\newline
We call $\mathbf{W}$ \emph{suitably indexed} if for $u\in c_i,v\in c_j$ it holds that $i<j$ implies $u<v$. In that case the index set is ordered block wise and $\mathbf{W}$ is explicitly block diagonal. In order to simplify the exposition, we may w.l.o.g. assume a suitable indexing.
\begin{definition}\label{defFM}Let $\mathbf{A}\in\mathbb{C}^{N\times N}$ and let $\mathbf{W}$ be an admissible weighted indicator matrix and let $\alpha\in\mathbb{R}$. The \emph{generalized quotient} $\mathbf{E}^{\alpha}$ is given by
$$\mathbf{E}^{\alpha}=\left(\mathbf{W}'\mathbf{W}\right)^{-\frac{1-\alpha}{2}}\mathbf{W}'
\mathbf{A}\mathbf{W}\left(\mathbf{W}'\mathbf{W}\right)^{-\frac{1+\alpha}{2}}.$$
\end{definition}
We call $\mathbf{E}^{0}$ the \emph{Rayleigh quotient}. The matrix entries of $\mathbf{E}^{\alpha}$ are
\begin{equation}
e^{\alpha}_{ij}=\left(\frac{1}{\lVert\mathbf{w}_i\rVert}\right)^{\left(1-\alpha\right)}
\mathbf{w}_i'\mathbf{A}_{ij}\mathbf{w}_j
\left(\frac{1}{\lVert\mathbf{w}_j\rVert}\right)^{\left(1+\alpha\right)}.
\end{equation}
Since $\mathbf{E}^{\alpha}=\left(\mathbf{W}\mathbf{W}\right)^{\frac{\alpha}{2}}\mathbf{E}^0\left(\mathbf{W}\mathbf{W}\right)^{-\frac{\alpha}{2}}$, all generalized quotients are similar. We distinguish the \emph{front quotient} $\mathbf{E}^{-}=\mathbf{E}^{-1}$ and the \emph{rear quotient} $\mathbf{E}^{+}=\mathbf{E}^{1}$. The matrix $\mathbf{A}$ is called \emph{front equitable} w.r.t. $\mathbf{W}$ if and only if
\begin{equation}
\mathbf{A}\mathbf{W}=\mathbf{W}\mathbf{E}^{-}
\quad\text{, i.e.}\quad
\forall\ i,j\in\left\{1,\ldots,k\right\}\ \mathbf{A}_{ij}\mathbf{w}_{j}=e^{-}_{ij}\mathbf{w}_{i}
\end{equation}
and we call $\mathbf{A}$ \emph{rear equitable} w.r.t. $\mathbf{W}$ if and only if
\begin{equation}
\mathbf{W}'\mathbf{A}={\mathbf{E}^{+}}\mathbf{W}'
\quad\text{, i.e.}\quad
\forall\ i,j\in\left\{1,\ldots,k\right\}\ \mathbf{w}_{i}'\mathbf{A}_{ij}={e^{+}_{ij}}\mathbf{w}_{j}'.
\end{equation}
\begin{definition}
Maintaining the notation above the \emph{front} and \emph{rear deviation vectors} are defined respectively as
$$\mathbf{t}^{-}_{ij}=\frac{1}{\lVert\mathbf{w}_j\rVert}\left(\mathbf{A}_{ij}\mathbf{w}_{j}-e^{-}_{ij}\mathbf{w}_{i}\right)
\quad\text{ and }\quad
{\mathbf{t}^{+}_{ij}}=\frac{1}{\lVert\mathbf{w}_i\rVert}\left({\mathbf{w}_{i}'\mathbf{A}_{ij}}-{e^{+}_{ij}}\mathbf{w}_{j}'\right)',$$
and the \emph{front} and \emph{rear deviation matrices} are
$$\mathbf{T}^{\pm}=\left(\begin{array}{ccc}
\mathbf{t}^{\pm}_{11}&\cdots&\mathbf{t}^{\pm}_{1k}\\
\vdots&\ddots&\vdots\\
\mathbf{t}^{\pm}_{k1}&\cdots&\mathbf{t}^{\pm}_{kk}
\end{array}\right)\in\mathbb{C}^{N\times k}\ \ ,\textrm{i.e.}\ \ 
\begin{array}{lll}
\mathbf{T}^{-}=\left(\mathbf{A}\mathbf{W}-\mathbf{W}\mathbf{E}^{-}\right)\left(\mathbf{W}'\mathbf{W}\right)^{-\frac{1}{2}}\\
\\
\mathbf{T}^{+}=
\left(\mathbf{A}'\mathbf{W}-\mathbf{W}{\mathbf{E}^{+}}'\right)\left(\mathbf{W}'\mathbf{W}\right)^{-\frac{1}{2}}.
\end{array}$$
\end{definition}
The entries of $\mathbf{E}^{\pm}$ and the deviation vectors have an intuitive interpretation in the framework of ordinary equitability arising for $\mathbf{w}_i=\mathbf{j}_{n_i}$. Then $e^{-}_{ij}$ and $\lVert\mathbf{t}^{-}_{ij}\rVert$ ($e^{+}_{ij}$ and $\lVert\mathbf{t}^{+}_{ij}\rVert$) are the mean and the standard deviation of the row (column) sums of $\mathbf{A}_{ij}$.\newline
Scaling the vector blocks $\mathbf{w}_{i}$ by $\mu_i\in\mathbb{C}\setminus\left\{0\right\}$ changes the entries of the generalized quotient to $\mu_i^{\alpha}e^{\alpha}_{ij}\mu_j^{-\alpha}$ although such a transformation sustains equitability (if present). Note that $e^{0}_{ij}$ and $\lVert{\mathbf{t}^{\pm}_{ij}}\rVert$, and therefore the singular values of $\mathbf{T}^{\pm}$, are independent of such a scaling. By definition, $\mathbf{T}^{\pm}$ is an all zero matrix if and only if its respective equitability holds. At the end of this section, we will consider suitable norms of $\mathbf{T}^{\pm}$ as measures for deviation from equitability.
\subsection{(Approximate) Block Triangularization}
Let $\mathbf{W}$ be an admissible weighted indicator matrix of a partition $\Pi=\left(c_1,\ldots,c_k\right)$ with weight vector $\mathbf{w}\in\mathbb{C}^{N}$ and indicator matrix $\mathbf{B}$. Let $\mathbf{w}_i$ be induced by $c_i$. Replacing $\mathbf{w}_i$ by $\mathbf{f}_{n_i}$ for all $i$ yields the new vector $\mathbf{f}$. Let $\mathbf{N}=\left(\mathbf{W}'\mathbf{W}\right)^{\frac{1}{2}}$ and
let $\mathbf{V}=\operatorname{diag}\left(\beta_1,\ldots,\beta_k\right)$ be a unitary diagonal matrix of size $k$. We introduce
\begin{equation}
\mathbf{Y}\left(\mathbf{W},\mathbf{V}\right)=\mathbf{Y}\left(\mathbf{w},\Pi,\mathbf{V}\right)=\operatorname{diag}\left(\mathbf{w}\right)\mathbf{B}
-\operatorname{diag}\left(\mathbf{f}\right)\mathbf{B}\mathbf{N}\mathbf{V}',
\end{equation}
which has the form of a weighted indicator matrix. The actual choice of $\mathbf{V}$ is a priori arbitrary. This freedom may be exploited in order to enhance the numerical properties of the transformation matrix given in the next definition.
\begin{definition}\label{defH}
Let $\mathbf{Y}$ be derived from an admissible weighted indicator matrix $\mathbf{W}$ and a unitary diagonal matrix $\mathbf{V}$ as above, then
$$\mathbf{H}\left(\mathbf{W},\mathbf{V}\right)=\mathbf{I}_N-\mathbf{Y}\left(\mathbf{W}'\mathbf{Y}\right)^{\dagger}\mathbf{Y}'.$$
\end{definition}
Since $\mathbf{Y}$ and $\mathbf{W}$ have the same block diagonal form, $\mathbf{Y}'\mathbf{W}$ is a diagonal matrix and $\mathbf{H}\left(\mathbf{V},\mathbf{W}\right)$ is block diagonal, hence its numerical properties are comparable to those of a single Householder matrix. In particular, the costs for computing and storing are of order $O\left(N\right)$, and it can be applied to a square matrix in $O\left(N^2\right)$. For suitably indexed $\mathbf{W}$ the block diagonal form of $\mathbf{H}\left(\mathbf{W},\mathbf{V}\right)$ is explicit,
\begin{equation}
\mathbf{H}\left(\mathbf{W},\mathbf{V}\right)=\operatorname{diag}\left(\mathbf{H}\left(\mathbf{w}_1,\beta_1\right),\ldots,\mathbf{H}\left(\mathbf{w}_k,\beta_k\right)\right).
\end{equation}
The diagonal blocks are given in definition \eqref{def_H(x)}. For $\mathbf{A}\in\mathbb{C}^{N\times N}$ we consider
\begin{equation}
\mathbf{\tilde{A}}=\mathbf{H}\left(\mathbf{W},\mathbf{V}\right)'\mathbf{A}\mathbf{H}\left(\mathbf{W},\mathbf{V}\right)\quad\text{with}\quad\mathbf{\tilde{A}}_{ij}=\mathbf{H}\left(\mathbf{w}_i,\beta_i\right)'\mathbf{A}_{ij}\mathbf{H}\left(\mathbf{w}_j,\beta_j\right).
\end{equation}
By the properties \eqref{eq_relations_H} of the $\mathbf{H}\left(\mathbf{w}_i,\beta_i\right)$ it follows that
\begin{align}
\mathbf{\tilde{A}}_{ij}\mathbf{f}_{n_j}&\sim \mathbf{f}_{n_i}\ \forall i,j
\text{ if and only if }\mathbf{A}\text{ is front equitable w.r.t. }\mathbf{W},\\
 \mathbf{f}_{n_i}'\mathbf{\tilde{A}}_{ij}&\sim \mathbf{f}_{n_j}'\ \forall i,j
\text{ if and only if }\mathbf{A}\text{ is rear equitable w.r.t. }\mathbf{W}.
\end{align}
If we consider for a moment front (row) equitability, the first column (row) of each block $\mathbf{\tilde{A}}_{ij}$ would be all zero from its second to last entry. This implies an implicit block triangular form of $\mathbf{\tilde{A}}$, which can be made explicit by the following permutation mapping the first index of each cell accordingly into $\left\{1,\ldots,k\right\}$.
\begin{definition}\label{defOmega}
Let $\mathbf{n}=\left(n_1,\ldots,n_k\right)$ be a sequence of $k$ positive integers with $\sum_{i=1}^{k}n_i=N$. The permutation
$\Omega_{\mathbf{n}}: \left\{1,\ldots,N\right\}\to\left\{1,\ldots,N\right\}$ is defined by
$$\Omega_{\mathbf{n}}\left(m_i+\sum\limits_{j=i}^{i-1}n_j\right)
=\left\{
\begin{array}{ll}
i & ,m_i=1\\
k-i+m_i+\sum\limits_{j=i}^{i-1}n_j & ,m_i\in\left\{2,\ldots,n_i\right\}
\end{array}
\right.$$
with $i\in\left\{1,\ldots,k\right\}$.
\end{definition}
We proceed with the general case and give the following theorem, which may be seen as a corollary of theorem \eqref{theoGen}. In order to keep this section self contained, it is proven independently.
\begin{theorem}\label{theo1}
Let $\Pi=\left(c_1,\ldots,c_k\right)$ be an admissible partition for $\mathbf{A}\in\mathbb{C}^{N\times N}$ and $\mathbf{w}\in\mathbb{C}^{N}$ with weighted indicator matrix $\mathbf{W}\in\mathbb{C}^{N\times k}$, generalized quotient $\mathbf{E}^{\alpha}$ and deviation matrices $\mathbf{T}^{\pm}$. Let $\mathbf{V}=\operatorname{diag}\left(\beta_1,\ldots,\beta_k\right)$ be a unitary diagonal matrix and let $\mathbf{\tilde{H}}=\mathbf{H}\left(\mathbf{W},\mathbf{V}\right)$ as in definition \eqref{defH} and let $\mathbf{\Omega}$ be the permutation matrix corresponding to $\Omega_{\left(\left|c_1\right|,\ldots,\left|c_k\right|\right)}$. Let
$$\mathbf{\tilde{A}}=\mathbf{\tilde{H}}'\mathbf{A}\mathbf{\tilde{H}}=
\left(\begin{array}{ccc}
\mathbf{\tilde{A}}_{11}&\cdots&\mathbf{\tilde{A}}_{1k}\\
\vdots&\ddots&\vdots\\
\mathbf{\tilde{A}}_{k1}&\cdots&\mathbf{\tilde{A}}_{kk}\\
\end{array}\right)
,\quad
\mathbf{\tilde{A}}_{ij}=\mathbf{H}\left(\mathbf{w}_{i},\beta_i\right)'\mathbf{A}_{ij}\mathbf{H}\left(\mathbf{w}_{j},\beta_j\right)$$
and 
$$\mathbf{\hat{A}}=\mathbf{\Omega}'\mathbf{\tilde{A}}\mathbf{\Omega}=
\left(\begin{array}{cc}
\mathbf{E}^{\phantom{-}}&{\mathbf{D}^{+}}'\\
\mathbf{D}^{-}&\mathbf{F}^{\phantom{+}}
\end{array}\right)\quad\text{with}\quad
\mathbf{E}\in\mathbb{C}^{k\times k}.$$
Then $\mathbf{\hat{A}}$ is unitarily similar to $\mathbf{A}$, the upper left block $\mathbf{E}$ is unitarily similar to the Rayleigh quotient $\mathbf{E}^0$ and the off-diagonal blocks $\mathbf{D}^{\pm}$ have the same singular values as $\mathbf{T}^{\pm}$, respectively. Additionally, any eigenvector $\mathbf{\hat{z}}$ of $\mathbf{\hat{A}}$ yields an eigenvector $\mathbf{z}=\mathbf{\tilde{H}}\mathbf{\Omega}\mathbf{\hat{z}}$ of $\mathbf{A}$ to the same eigenvalue.
\end{theorem}
\begin{proof}
Unitary similarity to $\mathbf{A}$ follows from the unitarity of $\mathbf{\tilde{H}}$ and $\mathbf{\Omega}$.\newline
Considering the matrix blocks $\mathbf{\tilde{A}}_{ij}$ of $\mathbf{\tilde{A}}$ induced by cells $c_i$ and $c_j$ we have 
\begin{equation}
e_{ij}=\mathbf{f}_{n_i}'\mathbf{\tilde{A}}_{ij}\mathbf{f}_{n_j}=
\frac{\overline{\beta_i}}{\lVert\mathbf{w}_i\rVert}
\frac{\beta_j}{\lVert\mathbf{w}_j\rVert}
\mathbf{w}_{i}'\mathbf{A}_{ij}\mathbf{w}_{j}=
\frac{\beta_j}{\beta_i}e^{0}_{ij}.
\end{equation}
By $\mathbf{\Omega}$ those $e_{ij}$ are mapped accordingly into the upper left block $\mathbf{E}$. Therefore, we may rewrite $\mathbf{E}=\mathbf{V}'\mathbf{E}^0\mathbf{V}$, which proofs unitary similarity of $\mathbf{E}$ and $\mathbf{E}^{0}$.\newline
In order to show that $\mathbf{D}^{\pm}$ is unitarily equivalent to $\mathbf{T}^{\pm}$, we observe that by the properties of $\mathbf{\Omega}$ we can write $\mathbf{D}^{\pm}$ as
\begin{equation}
\mathbf{D}^{\pm}=\left(\begin{array}{ccc}
\mathbf{d}^{\pm}_{11}&\cdots&\mathbf{d}^{\pm}_{1k}\\
\vdots&\ddots&\vdots\\
\mathbf{d}^{\pm}_{k1}&\cdots&\mathbf{d}^{\pm}_{kk}\\
\end{array}\right)\in\mathbb{C}^{\left(N-k\right)\times k},
\end{equation}
wherein $\mathbf{d}^{-}_{ij}$ is the first column and ${\mathbf{d}^{+}_{ij}}'$ is the first row of the matrix block $\mathbf{\tilde{A}}_{ij}$ starting from the second entry. We have
\begin{equation}
\left(\begin{array}{c}\hspace{-0.5em}0\\\mathbf{d}^{-}_{ij}\end{array}\right)=
\mathbf{\tilde{A}}_{ij}\mathbf{f}_{n_j}-e^{-}_{ij}\mathbf{f}_{n_i}=
\mathbf{H}\left(\beta_i,\mathbf{w}_i\right)'\mathbf{t}^{-}_{ij},
\end{equation}
\begin{equation}
\left(0,{\mathbf{d}^{+}_{ij}}'\ \right)=
\mathbf{f}_{n_i}'\mathbf{\tilde{A}}_{ij}-e^{+}_{ij}\mathbf{f}_{n_j}'=
{\mathbf{t}^{+}_{ij}}'\mathbf{H}\left(\beta_j,\mathbf{w}_j\right),
\end{equation}
which shows that
\begin{equation}\label{OHT=0D}
\mathbf{\Omega}'\mathbf{\tilde{H}}'\mathbf{T}^{\pm}=
\left(\begin{array}{c}
\mathbf{0}^{\phantom{\pm}}\\\mathbf{D}^{\pm}
\end{array}\right).
\end{equation}
The eigenvector relation can be shown by applying $\mathbf{\tilde{H}}\mathbf{\Omega}$ from the left to
\begin{equation}
\lambda\mathbf{\hat{z}}=\mathbf{\hat{A}}\mathbf{\hat{z}}=
\mathbf{\Omega}'\mathbf{\tilde{H}}'\mathbf{A}\mathbf{z}.
\end{equation}
\end{proof}
\subsection{Deviation from Equitability}
Let $\lVert\cdot\rVert_{U}$ denote a unitarily invariant norm.
\begin{corollary}\label{cor1}
$$\lVert\mathbf{D}^{\pm}\rVert_{U}=\lVert\mathbf{T}^{\pm}\rVert_{U}.$$
\end{corollary}
\begin{corollary}\label{cor2} Let 
$\mathbf{T}^{-}_{\mathbf{\Theta}}=
\left(\mathbf{A}\mathbf{W}-
\mathbf{W}\mathbf{\Theta}\right)\mathbf{N}^{-1}$ and $\mathbf{T}^{+}_{\mathbf{\Theta}}=\mathbf{N}^{-1}\left(\mathbf{W}'\mathbf{A}-
\mathbf{\Theta}\mathbf{W}'\right)$ with $\mathbf{N}=\left(\mathbf{W}'\mathbf{W}\right)^{\frac{1}{2}}$. Then
$$\lVert\mathbf{T}^{\pm}\rVert_{U}=
\min_{\mathbf{\Theta}}\lVert\mathbf{T}^{\pm}_{\mathbf{\Theta}}\rVert_{U}.$$
The minimum is unique if $\lVert\cdot\rVert_{U}$ is a Schatten norm.
\end{corollary}
\begin{proof}Applying $\mathbf{\Omega}'\mathbf{\tilde{H}}'$ from the left and $\mathbf{V}'$ from the right to $\mathbf{T}^{-}_{\mathbf{\Theta}}$
yields
\begin{align}
\lVert\mathbf{\Omega}'\mathbf{\tilde{H}}'\mathbf{T}^{-}_{\mathbf{\Theta}}\mathbf{V}'\rVert_U&=
\lVert\mathbf{\bar{A}}
\mathbf{\Omega}'\mathbf{\tilde{H}}'
\mathbf{W}\mathbf{N}^{-1}\mathbf{V}-
\mathbf{\Omega}'\mathbf{\tilde{H}}'\mathbf{W}\mathbf{\Theta}\mathbf{N}^{-1}\mathbf{V}\rVert_U\nonumber\\
&=\lVert\left(\begin{array}{c}\mathbf{E}^{\phantom{-}}\\\mathbf{D}^{-}\end{array}\right)
-\left(\begin{array}{c}
\mathbf{V}'\mathbf{N}
\mathbf{\Theta}\mathbf{N}^{-1}\mathbf{V}\\
\mathbf{0}\end{array}\right)\rVert_U
\end{align}
using $\mathbf{\Omega}'\mathbf{\tilde{H}}'\mathbf{W}=
\left(\begin{array}{c}
\mathbf{V}'\mathbf{N}\\
\mathbf{0}\end{array}\right)
$.
The last term is readily minimized for
\begin{equation}
\Theta=\mathbf{N}^{-1}\mathbf{V}\mathbf{E}\mathbf{V}'\mathbf{N}=\left(\mathbf{W}'\mathbf{W}\right)^{-\frac{1}{2}}\mathbf{E}^{0}\left(\mathbf{W}'\mathbf{W}\right)^{\frac{1}{2}}=\mathbf{E}^{-}.
\end{equation}
Obviously, the minimization is unique for several $\lVert\cdot\rVert_{U}$ including the Schatten norms. A similar proof applies for $\mathbf{T}^{+}$.
\end{proof}
The idea underlying the last proof is essentially the same as in \cite[proof~of~theorem~11]{Dax20101234}.
A particular useful choice for ${\lVert\cdot\rVert_{U}}$ might be the Frobenius norm, which upper bounds the spectral norm. Its square is simply the sum of the squared norms of the deviation vectors. One may also think of other characterizations for approximate equitable partitions which have moderate computational costs, for instance the number of nonzero columns of $\mathbf{T}^{\pm}$, which upper bounds the rank.
\section{Discussion and Remarks}\label{secDMR}
\subsection{Relating Equitability Deviation and Spectral Deviation}
Since $\mathbf{\hat{A}}$ and $\mathbf{A}$ are unitarily similar and by corollaries \eqref{cor1} and \eqref{cor2} of theorem \eqref{theo1}, we may in a sense 'measure' the deviation of a partition from being equitable by using a suitable unitarily invariant norm of $\mathbf{T}^{\pm}$, yielding a norm of $\mathbf{D}^{\pm}$, which in turn may serve as a measure for the deviation of the joint eigenvalue sets or the joint singular value sets of $\mathbf{E}$ and $\mathbf{F}$ from the respective values of $\mathbf{A}$.\newline
As an example we consider the spectral norm and the eigenvalue bound of Weyl for Hermitian matrices. Assuming Hermiticity we may set
$\mathbf{D}^{\pm}=\mathbf{D}$ and 
\begin{equation}
\mathbf{\hat{A}}=\left(\begin{array}{cc}\mathbf{E}&\mathbf{0}\\
\mathbf{0}&\mathbf{F}\end{array}\right)+\left(\begin{array}{cc}\mathbf{0}&\mathbf{D}'\\
\mathbf{D}&\mathbf{0}\end{array}\right).
\end{equation}
Let $\mu_1\leq\ldots\leq\mu_N$ be the joint spectrum of Hermitian $\mathbf{E}$ and $\mathbf{F}$, $\lambda_1\leq\ldots\leq\lambda_N$ the eigenvalues of $\mathbf{A}$ and let $\tau_{\text{spec}}$ be the largest singular value of $\mathbf{D}$.
We have
\begin{equation}
\left|\mu_i-\lambda_i\right|\leq\tau_{\text{spec}}\ ,\ 1\leq i\leq N
\end{equation}
by the Weyl inequalities. Many more pertubation bounds on eigenvalues and singular values and thier corresponding vectors are feasible, e.g. \cite{EI98AbsPer},\cite{Deif1995403},\cite{E85OptBou}.

\subsection{Cognate Concepts}
The notion of quasi-block-stochastic matrices of Kuich~\cite{K68QuaBlo} as a generalization of quasi-stochastic matrices~\cite{H55QuaSto} bears a close resemblance to our notion of equitability. A minor difference is that for quasi-block-stochastic matrices it is required that the first entry of each $\mathbf{v}_i$ has to be $1$. Kuich also describes how to exploit this structure to triangularize a (real) matrix by a (real, in general not unitary) similarity transformation using a theorem of Haynsworth~\cite{H59AppThe}.
Another similar but less general concept is used by Fiol and Carriga and is called \emph{pseudo-regular} partitions. It considers a positive eigenvector $\mathbf{v}$ of binary matrices~\cite[pp. 278/9]{F99EigInt}. The partition $\Pi$ of the matrix is pseudo-regular if $\mathbf{v}$ and $\Pi$ induce a (weighted) equitable partition. Since $\mathbf{v}$ is fixed up to a positive scale factor, the pseudo-quotient (i.e. front divisor) is unique.\newline
There are some more techniques in network analysis which can be described as variations of \eqref{equiCentral} and which are used to partition the node set of a graph (=assigning roles) according to structural properties and to derive a smaller graph (the quotient or image graph) which gives a condensed representation of essential relations between the cells (=roles) of that partition.
Some of those are without apparent regard to the spectrum. For instance, Kate and Ravindran introduced \emph{epsilon equitable} partitions for (an adjacency matrix $\mathbf{A}$ of) a simple graph~\cite{KR09EpsEqu}. Let $\Pi=\left(c_1,\ldots,c_k\right)$ be a partition of the node set of $\mathbf{A}$. Let $\mathbf{A}_{ij}$ be induced by the $i$-th row cell and the $j$-th column cell. Let $\mathbf{r}_{ij}=\mathbf{A}_{ij}\mathbf{j}_{n_j}$ be a column vector of length $n_i=\left|c_i\right|$. If 
\begin{equation}
\forall\ i,j\in\left\{1,\ldots,k\right\}\quad\max\limits_{1\leq v,w\leq n_i} \left|\mathbf{r}_{ij,v}-\mathbf{r}_{ij,w}\right|\leq\epsilon
\end{equation}
then $\Pi$ is called $\epsilon$-equitable. The ordinary equitable partition arises for $\epsilon=0$.
Another variation of \eqref{equiCentral} can be employed to describe the concept of \emph{regular equivalence}~\cite{EB94RegEqu}, which is defined by the restriction that for a partition $\Pi$ any vector $\mathbf{r}_{ij}=\mathbf{A}_{ij}\mathbf{j}_{n_j}$ must have either no zero entry or all entries zero i.e. 
\begin{equation}
\forall\ i,j\in\left\{1,\ldots,k\right\}\quad\prod_v\mathbf{r}_{ij,v}=0\Rightarrow\sum_v\left|\mathbf{r}_{ij,v}\right|=0.
\end{equation}
\subsection{Finding Equitable Partitions}
There are several algorithms for finding ordinary equitable partitions of graphs and matrices, for instance ~\cite{B99CompEP},~\cite{GKMS14}.\newline 
We sketch the most often employed top-down approach made suitable to the case of finding an ordinary front equitable partition of a complex matrix $\mathbf{A}$. At each step one considers a temporary partition (initially often the single cell partition) and (sequentially) subdivides any cell $c_i$ for any $j$ according to the entries of $\mathbf{A}_{ij}\mathbf{j}_{n_j}$, called colors, s.t. each subcell is induced by a unique color, until this subdivision is non trivial, resulting in a refined partition. One iterates until any feasible subdivision is trivial, i.e. the final partition is the unique coarsest front equitable refinement (w.r.t. to the initial partition). Of course, this can be adapted for the weighted case. However, if the weight vector $\mathbf{w}$ has no zero entries one may employ the sketched procedure for the unweighted case readily by considering the matrix $\operatorname{diag}\left(\mathbf{w}\right)^{-1}\mathbf{A}\operatorname{diag}\left(\mathbf{w}\right)$. This follows by left multiplication of $\operatorname{diag}\left(\mathbf{w}\right)^{-1}$ to the equitability condition
\begin{equation}
\mathbf{A}\operatorname{diag}\left(\mathbf{w}\right)\mathbf{B}=\operatorname{diag}\left(\mathbf{w}\right)\mathbf{B}\mathbf{E}^{-}.
\end{equation}
In general, the choice of a weight vector $\mathbf{w}$ may be guided by insights into the problem underlying the considered matrix $\mathbf{A}$. In search for $\mathbf{w}$, one may also exploit that the columns of the weighted indicator matrix $\mathbf{W}$ are a basis for the linear span of all eigenvectors of $\mathbf{A}$ corresponding to eigensolutions of $\mathbf{E}^{-}$. As an example, let $\mathbf{x}$ and $\mathbf{y}$ be two such eigenvectors for different eigenvalues. Considering them separately using the top down approach above one finds the single cell partition since $\mathbf{x}$ and $\mathbf{y}$ are eigenvectors. This may be avoided by using a non trivial linear combination, which lies in the linear span of the columns of $\mathbf{W}$ but is not an eigenvector.\newline
How to find partitions with suitably small but non zero deviation from equitability is out of the scope of this article.
%%%%%%%%%%%%%%%%%%%%%%%%%%%%%%%%%%%%%%%%%%%%%%%%%%%%%%%%%%%%%%%%%%%%%%%%%%%%%%%%%%
\appendix
\addtocontents{toc}{\protect\contentsline{chapter}{Appendix:}{}}
\section{Generalization as a Singular Value Decomposition}
Our proposed method for block triangularization can be described as an employment of a one-step singular value decomposition (SVD) of the weighted indicator matrix $\mathbf{W}$ as
\begin{equation}
\mathbf{W}=\left[\mathbf{\tilde{H}}\mathbf{\Omega}\right]\left(\begin{array}{c}\mathbf{N}\\\mathbf{0}\end{array}\right)\mathbf{V}'
\end{equation}
wherein the square diagonal matrix $\mathbf{N}$ contains the singular values of $\mathbf{W}$ and $\mathbf{V}=\operatorname{diag}\left(\beta_1,\ldots,\beta_k\right)$ is unitary diagonal. In deed, if we interpret $\mathbf{f}_{n_i}$ and the vector blocks $\mathbf{w}_i$ as matrices in $\mathbb{C}^{n_i\times1}$, then $\mathbf{w}_i$ has the SVD
\begin{equation}
\mathbf{w}_i=
\mathbf{H}\left(\mathbf{w}_i,\beta_i\right)
\frac{\lVert\mathbf{w}_i\rVert}{\beta_i}\mathbf{f}_{n_i}=
\mathbf{H}\left(\mathbf{w}_i,\beta_i\right)
\left(\begin{array}{c}\lVert\mathbf{w}_i\rVert\\\mathbf{0}\end{array}\right)
\overline{\beta_i}.
\end{equation}
In that view, one may obtain a generalization by replacing the non vanishing vector blocks $\mathbf{w}_i$ by rectangular matrix blocks $\mathbf{W}_i$ with maximal column rank. In the remainder of this section we build on this idea and derive an approximate block triangularization of a rectangular matrix $\mathbf{A}\in\mathbb{C}^{m\times n}$, using given SVDs of a pair of block diagonal matrices with maximal column rank, acting on the rows and columns of $\mathbf{A}$ respectively and separately.\newline
For notational convenience we define a $2\times 1$ block matrix with empty lower block and the identity matrix in the square upper block.
\begin{definition}\label{defInr}
Let $r$ and $n$ be positive integers with $r\leq n$.
$$\mathbf{I}_n^r=
\left(\begin{array}{c}\mathbf{I}_r\\\mathbf{0}\end{array}\right)\in\left\{0,1\right\}^{n\times r}.$$
\end{definition}
We may identify $\mathbf{I}_{n}^1=\mathbf{f}_{n}$. As a block diagonal generalization we define
\begin{definition}\label{defInrBD} Let $\mathbf{r}=\left(r_1,\ldots,r_k\right)$ and $\mathbf{n}=\left(n_1,\dots,n_k\right)$ be ordered sequences of positive integers, s.t. $r_i\leq n_i$.
$$\mathbf{I}_{\mathbf{n}}^{\mathbf{r}}=\operatorname{diag}\left(\mathbf{I}_{n_1}^{r_1},\ldots,\mathbf{I}_{n_k}^{r_k}\right).$$
\end{definition}
We will also utilize the following permutation.
\begin{definition}\label{defOmGen} Let $\mathbf{r}=\left(r_1,\dots,r_k\right)$ and $\mathbf{n}=\left(n_1,\dots,n_k\right)$ be ordered sequences of positive integers s.t. $\forall\ i\in\left\{1,\ldots,k\right\}\ r_i\leq n_i$. Let $r=\sum_ir_i$ and $n=\sum_in_i$. Then $\Omega_{\mathbf{n}}^{\mathbf{r}}: \left\{1,\ldots,n\right\}\to\left\{1,\ldots,n\right\}$ is defined by
$$\Omega_{\mathbf{n}}^{\mathbf{r}}\left(s_i+\sum\limits_{j=1}^{i-1}n_j\right)=\left\{\begin{array}{cc}
s_i+\sum\limits_{j=1}^{i-1}r_j&,0<s_i\leq r_i\\
s_i+r+\sum\limits_{j=1}^{i-1}\left(n_j-r_j\right)&,r_i<s_i\leq n_i\end{array}\right.$$
with $i\in\left\{1,\ldots,k\right\}$.
\end{definition}
$\Omega_{\mathbf{n}}^{\mathbf{r}}$ maps the first $r_i$ elements of cell $i$ into the first $r$ elements.
\begin{proposition}\label{PropOmega} In the notation of definitions \eqref{defInrBD} and \eqref{defOmGen} above, let 
$\mathbf{\Omega}$ be the permutation matrix corresponding to $\Omega_{\mathbf{n}}^{\mathbf{r}}$. Then $\mathbf{\Omega}'\mathbf{I}_{\mathbf{n}}^{\mathbf{r}}=\mathbf{I}_n^r.$
\end{proposition}
\begin{definition}\label{defAllGen}
Let $\mathbf{W}^{-}\in\mathbb{C}^{m\times q}$ be a block diagonal matrix with $l$ diagonal blocks $\mathbf{W}^{-}_i\in\mathbb{C}^{m_i\times q_i}$ of rank $q_i$ and let $\mathbf{W}^{+}\in\mathbb{C}^{n\times r}$ be a block diagonal matrix with $k$ diagonal blocks $\mathbf{W}^{+}_i\in\mathbb{C}^{n_i\times r_i}$ each of rank $r_i$.
Let singular value decompositions for the $\mathbf{W}^{-}_i$ be given by 
$$\mathbf{W}^{-}_i=\mathbf{U}^{-}_i\mathbf{S}^{-}_i{\mathbf{V}^{-}_i}'=
\mathbf{U}^{-}_i\left(\mathbf{I}_{m_i}^{q_i}\mathbf{N}^{-}_i\right){\mathbf{V}^{-}_i}'$$
with square unitary $\mathbf{U}^{-}_i\in\mathbb{C}^{m_i\times m_i}$ and $\mathbf{V}^{-}_i\in\mathbb{C}^{q_i\times q_i}$, and with $\mathbf{S}^{-}_i=\mathbf{I}_{m_i}^{q_i}\mathbf{N}^{-}_i$ wherein $\mathbf{N}^{-}_i\in\mathbb{R}^{q_i\times q_i}$ is a square diagonal matrix with positive diagonal elements. 
Let $\mathbf{U}^{-}\in\mathbb{C}^{m\times m}$, $\mathbf{S}^{-}\in\mathbb{C}^{m\times q}$, $\mathbf{N}^{-}\in\mathbb{R}^{q\times q}$, and $\mathbf{V}^{-}\in\mathbb{C}^{q\times q}$ be block diagonal with $l$ diagonal blocks given by $\mathbf{U}^{-}_i$, $\mathbf{S}^{-}_i$, $\mathbf{N}^{-}_i$, and $\mathbf{V}^{-}_i$, respectively. Let $\mathbf{\Omega
 }^{-}$ be the permutation matrix corresponding to the permutation $\Omega^{\left(q_1,\dots,q_l\right)}_{\left(m_1,\dots,m_l\right)}$ s.t. ${\mathbf{\Omega}^{-}}'\mathbf{S}^{-}=\mathbf{I}_m^q\mathbf{N}^{-}$.
This induces a singular value decomposition of $\mathbf{W}^{-}$ as $$\mathbf{W}^{-}=\mathbf{U}^{-}\left({\mathbf{\Omega}^{-}}'\mathbf{S}^{-}\right){\mathbf{V}^{-}}'=\mathbf{U}^{-}\left(\mathbf{I}_m^q\mathbf{N}^{-}\right){\mathbf{V}^{-}}'.$$
Let the corresponding relations hold for $\mathbf{W}^{+}$
and let $\mathbf{A}\in\mathbb{C}^{m\times n}$.
Define the \emph{Rayleigh quotient} $\mathbf{E}^{\mathrm{0}}\in\mathbb{C}^{q\times r}$ as a block matrix with 
\begin{equation}
\mathbf{E}^{0}=\left(\begin{array}{ccc}
\mathbf{E}^{0}_{11}&\cdots&\mathbf{E}^{0}_{1n}\\
\vdots&\ddots&\vdots\\
\mathbf{E}^{0}_{1m}&\cdots&\mathbf{E}^{0}_{mn}
\end{array}\right)
,\quad
\mathbf{E}^{0}_{ij}=
\mathbf{V}^{-}_{i}{\mathbf{I}^{q_i}_{m_i}}'{\mathbf{U}^{-}_{i}}'
\mathbf{A}_{ij}
\mathbf{U}^{+}_{j}\mathbf{I}^{r_j}_{n_j}{\mathbf{V}^{+}_{j}}'\in\mathbb{C}^{q_i\times r_j}.\nonumber
\end{equation}
The \emph{front} and {rear deviation matrices}, $\mathbf{T}^{-}\in\mathbb{C}^{m\times r}$ and $\mathbf{T}^{+}\in\mathbb{C}^{n\times q}$ respectively, are block matrices with 
\begin{align}
\mathbf{T}^{-}_{ij}=&
\phantom{(}\mathbf{A}_{ij}\phantom{)}\mathbf{U}^{+}_{j}\mathbf{I}^{r_j}_{n_j}{\mathbf{V}^{+}_{j}}'-
\mathbf{U}^{-}_{i}\mathbf{I}^{q_i}_{m_i}{\mathbf{V}^{-}_{i}}'\phantom{(}\mathbf{E}^{0}_{ij}\phantom{)'}\in\mathbb{C}^{m_i\times r_i},\nonumber\\
\mathbf{T}^{+}_{ij}=&
\left(\mathbf{A}_{ij}\right)'\mathbf{U}^{-}_{i}\mathbf{I}^{q_i}_{m_i}{\mathbf{V}^{-}_{i}}'-
\mathbf{U}^{+}_{j}\mathbf{I}^{r_j}_{n_j}{\mathbf{V}^{+}_{j}}'\left(\mathbf{E}^{0}_{ij}\right)'\in\mathbb{C}^{n_i\times q_i}.\nonumber
\end{align}
\end{definition}
\begin{proposition}
In the notation of definition \eqref{defAllGen}, $\mathbf{E}^{0}$, $\mathbf{T}^{-}$, and $\mathbf{T}^{+}$ are identical for all singular value decompositions of $\mathbf{W}^{-}$ and $\mathbf{W}^{+}$ which obey the block diagonal form.
\end{proposition}
\begin{proof} $\mathbf{E}^{0}$ and $\mathbf{T}^{\pm}$ can be entirely expressed in terms of $\mathbf{W}^{\pm}$ since
\begin{equation}
\mathbf{U}^{-}\mathbf{I}^{\mathbf{q}}_{\mathbf{m}}{\mathbf{V}^{-}}'=\mathbf{W}^{-}\left({\mathbf{W}^{-}}'\mathbf{W}^{-}\right)^{-\frac{1}{2}}
\text{and}\ \ 
\mathbf{U}^{+}\mathbf{I}^{\mathbf{r}}_{\mathbf{n}}{\mathbf{V}^{+}}'=\mathbf{W}^{+}\left({\mathbf{W}^{+}}'\mathbf{W}^{+}\right)^{-\frac{1}{2}}.
\end{equation}
\begin{align}
\mathbf{E}^{0}&=\left({\mathbf{W}^{-}}'\mathbf{W}^{-}\right)^{-\frac{1}{2}}{\mathbf{W}^{-}}'\mathbf{A}\mathbf{W}^{+}\left({\mathbf{W}^{+}}'\mathbf{W}^{+}\right)^{-\frac{1}{2}},
\\
\mathbf{T}^{-}&=
\mathbf{A}\mathbf{W}^{+}
\left({\mathbf{W}^{+}}'\mathbf{W}^{+}\right)^{-\frac{1}{2}}-
{\mathbf{W}^{-}}'\left({\mathbf{W}^{-}}'\mathbf{W}^{-}\right)^{-\frac{1}{2}}
\mathbf{E}^{0},
\\
\mathbf{T}^{+}&=
\mathbf{A}'\mathbf{W}^{-}
\left({\mathbf{W}^{-}}'\mathbf{W}^{-}\right)^{-\frac{1}{2}}-
{\mathbf{W}^{+}}'
\left({\mathbf{W}^{+}}'\mathbf{W}^{+}\right)^{-\frac{1}{2}}
{\left(\mathbf{E}^{0}\right)}'.
\end{align}
\end{proof}
\begin{theorem}\label{theoGen}
In the notation of definition \eqref{defAllGen} above, let
$$\mathbf{\hat{A}}={\mathbf{\Omega}^{-}}'{\mathbf{U}^{-}}'\mathbf{A}\mathbf{U}^{+}\mathbf{\Omega}^{+}=
\left(\begin{array}{cc}
\mathbf{E}^{\mathrm{\phantom{f}}}&{\mathbf{D}^{+}}'\\
\mathbf{D}^{-}&\mathbf{F}^{\mathrm{\phantom{r}}}
\end{array}\right)\quad\text{with}\quad\mathbf{E}\in\mathbb{C}^{q\times r}.$$
and let $\lVert\cdot\rVert_{U}$ be a unitarily invariant matrix norm. Then $\mathbf{E}$ and $\mathbf{E}^{0}$ are unitarily equivalent, $\mathbf{D}^{\pm}$ and $\mathbf{T}^{\pm}$ have the same singular values, respectively, and
\begin{align}
\lVert\mathbf{T}^{-}\rVert_{U}=&\min_{\mathbf{\Theta}}\lVert
\mathbf{A}\mathbf{U}^{+}\mathbf{I}^{\mathbf{r}}_{\mathbf{n}}{\mathbf{V}^{+}}'-
\mathbf{U}^{-}\mathbf{I}^{\mathbf{q}}_{\mathbf{m}}{\mathbf{V}^{-}}'\mathbf{\Theta}\rVert_{U}\nonumber\\
\lVert\mathbf{T}^{+}\rVert_{U}=&\min_{\mathbf{\Theta}}\lVert
\mathbf{A}'\mathbf{U}^{-}\mathbf{I}^{\mathbf{q}}_{\mathbf{m}}{\mathbf{V}^{-}}'-
\mathbf{U}^{+}\mathbf{I}^{\mathbf{r}}_{\mathbf{n}}{\mathbf{V}^{+}}'\mathbf{\Theta}'
\rVert_{U}.\nonumber
\end{align}
\end{theorem}
\begin{proof}
Unitary equivalence of $\mathbf{E}$ and $\mathbf{E}^{0}$ follows from
\begin{equation}
\mathbf{E}={\mathbf{I}^{q}_{m}}'\mathbf{\hat{A}}{\mathbf{I}^{r}_{n}}=
\operatorname{diag}\left(
\mathbf{U}^{-}_{1}\mathbf{I}^{q_1}_{m_1},\ldots,
\mathbf{U}^{-}_{l}\mathbf{I}^{q_l}_{m_l}\right)'
\mathbf{A}
\operatorname{diag}\left(
\mathbf{U}^{+}_{1}\mathbf{I}^{r_1}_{n_1},\ldots,
\mathbf{U}^{+}_{k}\mathbf{I}^{r_k}_{n_k}\right),
\end{equation}
which uses proposition \eqref{PropOmega}, yielding $\mathbf{E}^{0}=\mathbf{V}^{-}\mathbf{E}{\mathbf{V}^{+}}'.$\newline
That $\mathbf{D}^{-}$ and $\mathbf{T}^{-}$ share the same multiset of singular values follows from
\begin{align}
\left(\begin{array}{c}\mathbf{0}^{\phantom{-}}\\\mathbf{D}^{-}\end{array}\right)&=
\mathbf{\hat{A}}\mathbf{I}^{r}_{n}-\mathbf{I}^{q}_{m}\mathbf{E}=
{\mathbf{\Omega}^{-}}'{\mathbf{U}^{-}}'\mathbf{A}
\mathbf{U}^{+}\mathbf{\Omega}^{+}\mathbf{I}^{r}_{n}-
\mathbf{I}^{q}_{m}{\mathbf{V}^{-}}'\mathbf{E}^{0}{\mathbf{V}^{+}}\nonumber\\
&={\mathbf{\Omega}^{-}}'{\mathbf{U}^{-}}'\mathbf{T}^{-}{\mathbf{V}^{+}}.
\end{align}
The proof for $\mathbf{D}^{+}$ and $\mathbf{T}^{+}$ is analogous.\newline
Applying the unitary matrices ${\mathbf{\Omega}^{-}}'{\mathbf{U}^{-}}'$ from the left and ${\mathbf{V}^{+}}$ from the right to the second term in the penultimate equation of theorem \eqref{theoGen} yields
\begin{align}
\min_{\mathbf{\Theta}}
\lVert\mathbf{\tilde{A}}
{\mathbf{\Omega}^{+}}'\mathbf{I}^{\mathbf{r}}_{\mathbf{n}}-
{\mathbf{\Omega}^{-}}'\mathbf{I}^{\mathbf{q}}_{\mathbf{m}}
{\mathbf{V}^{-}}'\mathbf{\Theta}\mathbf{V}^{+}\rVert_U
=\lVert\left(\begin{array}{c}\mathbf{E}^{\phantom{-}}\\\mathbf{D}^{-}\end{array}\right)
-\left(\begin{array}{c}
{\mathbf{V}^{-}}'
\mathbf{\Theta}\mathbf{V}^{+}\\
\mathbf{0}\end{array}\right)\rVert_U.
\end{align}
The last term is readily minimized for $\mathbf{\Theta}={\mathbf{V}^{-}}\mathbf{E}{\mathbf{V}^{+}}'=\mathbf{E}^{0}.$
The minimum is obviously unique if $\lVert\cdot\rVert_U$ is a Schatten norm. A similar proof applies for the minimum property of $\mathbf{T}^{+}$.
\end{proof}
%%%%%%%%%%%%%%%%%%%%%%%%%%%%%%%%%%%%%%%%%%%%%%%%%%%%%%%%%%%%%%%%%%%%%%%%%%%%%%%%%%%%%%%%%%%
\bibliography{MTh_EEPfEBT_ref}
\end{document}